\documentclass[10pt]{IEEEtran}
\onecolumn
\usepackage[ruled,vlined]{algorithm2e}
\usepackage{url}

\usepackage{array,graphicx,lineno}
\usepackage{amsfonts}
\usepackage[cmex10]{amsmath}
\usepackage{amssymb}
\usepackage{booktabs,mathrsfs,bbm,amssymb,balance}
\usepackage{booktabs,mathrsfs,bbm,setspace}
\usepackage[table]{xcolor}

\def\bsl{\mathbf{A}}
\def\bsls{\mathbf{A}}

\usepackage[T1]{fontenc}

\usepackage{amsthm}
\newtheorem{theorem}{Theorem}
\newtheorem{definition}{Definition}
\newtheorem{lemma}{Lemma}
\newtheorem{remark}{Remark}
\usepackage[T1]{fontenc}
\usepackage[utf8]{inputenc}
\usepackage[table]{xcolor}

\newcommand{\R}{\mathbb{R}}
\newcommand{\Z}{\mathbb{Z}}
\newcommand{\sinc}{\operatorname{sinc}}

\def\DE{\stackrel{\mathrm{def}}{=}}

\def\j{\jmath}

\def\ind{\chi}

\def\l{\left (}
\def\r{\right )}

\def\SC{*_{\mathbf{A}}}
\def\SDC{\star_{\bsls}}

\newcommand{\nvec}[1]{ {\underline{#1}}}
\newcommand{\Lu}[1]{\lambda_{\mathbf{A}} \left( #1 \right)}
\newcommand{\Ld}[1]{{\overline\lambda}_{\mathbf{A}} \left( #1 \right)}

\newcommand{\Nu}[1]{\eta_{\mathbf{A}} \left( #1 \right)}
\newcommand{\Nd}[1]{{\overline\eta}_{\mathbf{A}} \left( #1 \right)}
\newcommand{\dtsaft}[1]{\widehat{#1}_\mathbf{A}\left( \omega \right)}

\newcommand{\up}[2]{ {\overset{\lower0.5em\hbox{$\smash{\scriptscriptstyle\rightharpoonup}$}} {{#1}}} \left( {#2} \right)}
\newcommand{\dn}[2]{ {\overset{\lower0.5em\hbox{$\smash{\scriptscriptstyle  \leftharpoonup}$}} {{#1}}} \left( {#2} \right)}

\newcommand{\DEq}[1]{\stackrel{(\ref{#1})}{=}}
\newcommand{\DTo}[1]{\stackrel{(\ref{#1})}{\to}}

\DeclareSymbolFont{symbols}{OMS}{fncsy}{m}{n}
\DeclareMathAlphabet{\mathcal}{OMS}{fncm}{m}{n}

\usepackage{url}
\DeclareSymbolFontAlphabet{\Bbb}{AMSb}

\begin{document}

\par\noindent {\LARGE\bf Shift--Invariant and Sampling Spaces Associated with the \\ Special Affine Fourier Transform \par}


{\vspace{4mm}\par\noindent
Ayush Bhandari$^\dag$ and Ahmed I. Zayed$^\ddag$
\par\vspace{2mm}\par}

{\vspace{2mm}\par\it
\noindent $^\dag$Media Laboratory, Massachusetts Institute of Technology\\
$\phantom{^{\dag}}$Cambridge, MA 02139--4307 USA. \par}

{\vspace{2mm}\par\it
\noindent $^\ddag$Department of Mathematical Sciences, \\ 
$\phantom{^{\ddag}}$DePaul University, Chicago, IL 60614-3250\par}

{\vspace{4mm}\par\noindent $\phantom{^{\dag,\ddag}}$\rm 
Email: \textrm{ayush@MIT.edu $\bullet$ \ azayed@condor.depaul.edu } \par}

{\vspace{7mm}\par\noindent\hspace*{5mm}\parbox{150mm}{\small \textbf{Abstract: }
The Special Affine Fourier Transformation or the SAFT generalizes a number of well known unitary transformations as well as signal processing and optics related mathematical operations. Shift-invariant spaces also play an important role in sampling theory, multiresolution analysis, and many other areas of signal and image processing. Shannon's sampling theorem, which is at the heart of modern digital communications,  is a special case of sampling in shift-invariant spaces. Furthermore, it is well known that the Poisson summation formula is equivalent to the sampling theorem and that the Zak transform is closely connected to the sampling theorem and the Poisson summation formula. These results have been known to hold in the Fourier transform domain for decades and were recently shown to hold in the Fractional Fourier transform domain by A. Bhandari and A. Zayed.

The main goal of this article is to show that these results also hold true in the SAFT domain. We provide a short, self--contained proof of Shannon's theorem for functions bandlimited in the SAFT  domain and then show that sampling in the SAFT domain is equivalent to orthogonal projection of functions onto a subspace of bandlimited basis associated with the SAFT domain. This interpretation of sampling leads to least--squares optimal sampling theorem. Furthermore, we show that this approximation procedure is linked with convolution and semi--discrete convolution operators that are associated with the SAFT domain. We conclude the article with an application of fractional delay filtering of SAFT bandlimited functions.}\par\vspace{3mm}}


\tableofcontents
 \begin{spacing}{1.5}
\newpage


\section{Introduction}

The Special Affine Fourier Transformation (SAFT), which was introduced in \cite{Abe1994}, is an integral transformation associated with a general inhomogeneous lossless
linear mapping in phase-space  that depends on six parameters independent of the phase-space coordinates. It maps the position $x$ and the wave number $k$ into
 \begin{equation}
\left[ {\begin{array}{*{20}{c}}
  {x'} \\
  {k'}
\end{array}} \right] = \left[ {\begin{array}{*{20}{c}}
  a&b \\
  c&d
\end{array}} \right]\left[ {\begin{array}{*{20}{c}}
  x \\
  k
\end{array}} \right] + \left[ {\begin{array}{*{20}{c}}
  p \\
  q
\end{array}} \right]
\label{SAFT}
 \end{equation}
with
\begin{equation}ad-bc=1.
\label{SAFT2}
\end{equation}
This transformation, which can model many general optical systems \cite{Abe1994},  maps any convex body into another convex body and (\ref{SAFT2}) guarantees that the area of the body is preserved by the transformation.
Such transformations form the inhomogeneous special linear group ISL$(2,\R).$
%
\begin{table}[!b]
\centering
\caption{ \textsf{SAFT, Unitary Transformations and Operations}}
\begin{tabular*}{0.75\textwidth}{p{5cm} p{6.15cm}}
\toprule
%
\rowcolor{yellow!35} 
\hline
SAFT Parameters $\left(\bsls \right) $ & Corresponding Unitary Transform \\
\hline
\addlinespace
$\bigl[ \begin{smallmatrix} a &b  & \vline & & {0}   \\
 c& d & \vline& &{0} \end{smallmatrix} \bigr] = \bsls_\textsf{LCT} $				&  \textbf{Linear Canonical Transform} \\	[2.5pt]

$\bigl[ \begin{smallmatrix} &\cos\theta&\sin\theta  & \vline & & {p}   \\
 -&\sin\theta&\cos\theta & \vline& &{q} \end{smallmatrix} \bigr] = {\bsls}_\theta^O $				&  \textbf{Offset Fractional Fourier Transform} \\	[2.5pt] 

$\bigl[ \begin{smallmatrix} &\cos\theta&\sin\theta  & \vline & & {0}   \\
 -&\sin\theta&\cos\theta & \vline& &{0} \end{smallmatrix} \bigr] = {\bsls}_\theta $				&  \textbf{Fractional Fourier Transform} \\	[2.5pt] 

$\bigl[ \begin{smallmatrix} &0 &1& \vline & & {p}  \\-&1&0 & \vline & & {q} \end{smallmatrix} \bigr] 
= {\bsls}_\text{FT}^O$  													& \textbf{Offset Fourier Transform (FT)}  \\		[2.5pt] 

$\bigl[ \begin{smallmatrix} &0 &1& \vline & & {0}  \\-&1&0 & \vline & & {0} \end{smallmatrix} \bigr] 
= {\bsls}_\text{FT}$  													& \textbf{Fourier Transform (FT)}  \\		[2.5pt] 
$\bigl[ \begin{smallmatrix} 0 & \j & \vline & & {0} \\ \j & 0 & \vline & & {0} \end{smallmatrix} \bigr] 
= {\bsls}_\text{LT}$ 													& \textbf{Laplace Transform (LT)}  \\	[2.5pt] 
$\bigl[ \begin{smallmatrix} \j \cos\theta &\j \sin\theta & \vline & & {0}  \\
 \j \sin \theta & -\j\cos\theta & \vline & & {0} \end{smallmatrix} \bigr] $  								& \textbf{Fractional Laplace Transform}   \\ 	[2.5pt] 
$\bigl[ \begin{smallmatrix} 1 & b & \vline & & {0}  \\0&1 & \vline & & {0} \end{smallmatrix} \bigr]$  						& \textbf{Fresnel Transform}  \\	[2.5pt] 
$\bigl[ \begin{smallmatrix} 1&\jmath b & \vline & & {0}  \\  \jmath&1 & \vline & & {0} \end{smallmatrix} \bigr]$				& \textbf{Bilateral Laplace Transform}  \\	[2.5pt] 
$\bigl[ \begin{smallmatrix} 1&-\jmath b & \vline & & {0}  \\ 0&1 & \vline & & {0} \end{smallmatrix} \bigr]$, $b \ge 0$  		& \textbf{Gauss--Weierstrass Transform}   \\	[2.5pt] 
$\tfrac{1}{{\sqrt 2 }} \bigl[\begin{smallmatrix} 0 & e^{ - {{\jmath\pi } 
\mathord{\left/{\vphantom {{j\pi } 2}} \right.\kern-\nulldelimiterspace} 2}} & \vline & & {0} 
\\-e^{ - {{\jmath\pi } \mathord{\left/{\vphantom {{j\pi } 2}} \right.\kern-\nulldelimiterspace} 2}} 
&1 & \vline & & {0} \end{smallmatrix} \bigr]$  													& \textbf{Bargmann Transform} \\[2.5pt] 
%
\addlinespace
\hline
\rowcolor{yellow!35} 
SAFT Parameters $\left(\bsls \right) $ & Corresponding Signal Operation \\
\hline
\addlinespace
$\bigl[ \begin{smallmatrix} 1/\alpha& 0  & \vline & & {0}   \\
 0 & \alpha & \vline& &{0} \end{smallmatrix} \bigr] = {\bsls}_\alpha $				&  \textbf{Time Scaling} \\	[2.5pt] 

$\bigl[ \begin{smallmatrix} 1 & 0  & \vline & & {\tau}   \\
 0 & 1 & \vline& &{0} \end{smallmatrix} \bigr] = {\bsls}_\tau $				&  \textbf{Time Shift} \\	[2.5pt]

$\bigl[ \begin{smallmatrix} 1 & 0  & \vline & & {0}   \\
 0 & 1 & \vline& &{\xi} \end{smallmatrix} \bigr] = {\bsls}_\xi $				&  \textbf{Frequency Shift} \\	[2.5pt] 

\addlinespace
\hline
\rowcolor{yellow!35} 
SAFT Parameters $\left(\bsls \right) $ & Corresponding Optical Operation \\
\hline
\addlinespace
$\bigl[ \begin{smallmatrix} &\cos\theta&\sin\theta  & \vline & & {0}   \\
 -&\sin\theta&\cos\theta & \vline& &{0} \end{smallmatrix} \bigr] = {\bsls}_\theta $			&  \textbf{Rotation} \\	[2.5pt] 

$\bigl[ \begin{smallmatrix} 1 & 0  & \vline & & {0}   \\
 \tau & 1 & \vline& &{0} \end{smallmatrix} \bigr] = {\bsls}_\tau $				&  \textbf{Lens Transformation} \\	[2.5pt] 
  
 $\bigl[ \begin{smallmatrix} 1 & \eta  & \vline & & {0}   \\
 0 & 1 & \vline& &{0} \end{smallmatrix} \bigr] = {\bsls}_\eta $				&  \textbf{Free Space Propagation} \\	[2.5pt] 
 
 $\bigl[ \begin{smallmatrix} e^{\beta} & 0  & \vline & & {0}   \\
 0 & e^{-\beta}  & \vline& &{0} \end{smallmatrix} \bigr] = {\bsls}_\beta $				&  \textbf{Magnification} \\	[2.5pt] 
 
  $\bigl[ \begin{smallmatrix} \cosh\alpha & \sinh\alpha  & \vline & & {0}   \\
\sinh\alpha & \cosh\alpha  & \vline& &{0} \end{smallmatrix} \bigr] = {\bsls}_\eta $				&  \textbf{Hyperbolic Transformation} \\	[2.5pt] 
\bottomrule
\end{tabular*}
\label{tab:1}
\end{table}
 The SAFT offers a  unified viewpoint of known signal processing transformations as well as optical operations on light waves. We have parametrically summarized these operations in Table~\ref{tab:1}.

The integral representation of the wave-function transformation linked with the transformation (\ref{SAFT}) and (\ref{SAFT2})
is given by,
\begin{align}
\label{SAFTD}
F(\omega)&=\hat{f}_{\bsls}(\omega)=\int_{\R} k(t,\omega)f(t)dt \qquad \l\mbox{SAFT of } f\l t \r \r\\
&= \frac{1}{\sqrt{2\pi |b|}}\int_{\R}
\exp\left\{\frac{j}{2b}\left( at^2+d\omega^2-2t\omega+2pt+2(bq-dp)\omega\right)\right\}f(t) dt, \notag
\end{align}
where $\bsls$  stands for the six parameters $(a,b,c,d,p, q),$ and
$$ k(t,\omega)=\frac{1}{\sqrt{2\pi |b|}}
\exp\left\{\frac{j}{2b}\left( at^2+d\omega^2-2t\omega+2pt+2(bq-dp)\omega\right)\right\}.$$

 The inversion formula for the SAFT is easily shown to be
 \begin{equation}
f(t)=\frac{1}{\sqrt{2\pi |b|}}\int_{\R}
F(\omega) \exp\left\{\frac{-j}{2b}\left( d\omega^2+at^2-2t\omega+ 2\omega(bq-dp)+2pt\right)\right\} dt\omega,
\end{equation}
which may be considered as the SAFT evaluated using
 ${\bf A}^{-1}$
where\footnote{With a little abuse of notation, we use $\nvec{\lambda} ^{ - 1}$ which should be understood as a parameter vector corresponding to the inverse--SAFT.},
\[{\bf A}^{-1} \DE \left[ \bsls^{-1} | \nvec{\lambda}^{-1} \right]\equiv \left[ {\begin{array}{*{20}{c}}
  { + d}&{ - b}&\vline & {bq - dp} \\
  { - c}&{ + a}&\vline & {cp - aq}
\end{array}} \right]\]
and to be precise,
\[ {\bf A}^{-1} = \left[ {\begin{array}{*{20}{c}}
  +d&{ - b} \\
  { - c}& + a
\end{array}} \right]{\text{ and }}{\nvec{\lambda} ^{ - 1}}  \DE  \left[ {\begin{array}{*{20}{c}}
  {bq - dp} \\
  {cp - aq}
\end{array}} \right].\]

 We also have $$  \langle f, g\rangle= \int_{\R} f(t)\overline{g}(t)dt=\int_{\R} F(\omega)\overline{G}(\omega) d\omega =\langle F, G\rangle ,$$ from which we obtain
 $\left\|f\right\|=\left\|F\right\|.$
When $p=0=q,$ we obtain the homogeneous special group  SL$(2,\R),$  which is
represented by the unimodular matrix
\[ \mathbf{M} = \left[ {\begin{array}{*{20}{c}}
  a&b \\
  c&d
\end{array}} \right].\]
The associated integral transform, which is called the {Linear Canonical Transform} (LCT), is given by
$$ F_{\mathsf{LCT}}(\omega)=\frac{1}{\sqrt{2\pi |b|}}\int_{\R}
\exp\left\{\frac{j}{2b}\left( at^2+ d\omega^2-2t\omega \right)\right\}f(t)  dt. $$

The linear canonical transform  has been used to solve problems in physics and quantum mechanics; see \cite{Moshinsky2003}. It includes several known transforms as special cases. For example, for $a=0=d, b=-1, c=1,$ we obtain the Fourier transform and for $a=\cos \theta=d, b=\sin \theta =-c,$ we obtain the fractional Fourier transform. The  Laplace, Gauss-Weierstrass, and Bargmann transforms are also special cases.
The inversion formula for the LCT is given by
\begin{equation}
f(t)=\frac{1}{\sqrt{2\pi |b|}}\int_{\R}
\exp\left\{\frac{-j}{2b}\left( d\omega^2+ at^2-2t\omega \right)\right\}F(\omega)  d\omega.
\end{equation}
If  the LCT of $f$ and $g$ are denoted by $F$ and $G,$ respectively, then Parseval's relation holds
$$\langle f, g \rangle=\int_{\R}f(x)\overline{g}(x) dx = \int_{\R}F(t)\overline{G}(t) dt=\langle F, G \rangle.$$
Let $$\mathbf{M}_1=\begin{bmatrix} a_1 & b_1\\c_1 & d_1\end{bmatrix} , \quad \mathbf{M}_2=\begin{bmatrix} a_2 & b_2\\c_2 & d_2\end{bmatrix} , $$
so that
$$
\mathbf{M}_{21} =   \begin{bmatrix} a_2 & b_2\\c_2 & d_2\end{bmatrix}\begin{bmatrix} a_1 & b_1\\c_1 & d_1\end{bmatrix}=  \begin{bmatrix} a_2a_1+ b_2c_1 &
a_2b_1 +b_2 d_1 \\c_2a_1+d_2c_1 & c_2b_1+ d_2d_1\end{bmatrix}.
$$
If the LCT corresponding to $\mathbf{M}_1,\mathbf{M}_2, \mathbf{M}_{21} $ are denoted by $ \mathcal{L}_1, \mathcal{L}_2, \mathcal{L}_{21},$ respectively, it can be shown that the composition
relation $\mathcal{L}_2\mathcal{L}_1=C \mathcal{L}_{21}, $ holds,
where $C$ is a constant. On the other hand, the composition relations is associative, that is,
 $$\left(\mathcal{L}_3\mathcal{L}_2\right)\mathcal{L}_1 =\mathcal{L}_3\left(\mathcal{L}_2\mathcal{L}_1 \right).$$

The analogue of Shannon sampling theorem for the fractional Fourier transform, the linear canonical transform, and the SAFT  were obtained in \cite{Zayed1996,Zayed1999,Zhao2009,Xia1996,Stern2007,Xiang2012,Xiang2013}. Although the sampling theorem can be easily obtained in a direct way, we will obtain it as a special case of more general results.

 Our goal is to extend key harmonic analysis results to the SAFT analogous to those for the FT and FrFT  and obtain the sampling theorem as a by-product. For example, in the Fourier transform domain, it is known that Shannon sampling theorem is a special case of sampling in shift-invariant spaces, as well as, sampling in reproducing-kernel Hilbert spaces.  Moreover, it is also known that the Poisson summation formula is equivalent to the sampling theorem and the Zak transform is closely connected to the sampling theorem and the Poisson summation formula. These results have been known to hold in the Fourier transform domain for decades and were recently shown to hold in the FrFT domain \cite{Bhandari2012}. The main goal of this article is to extend these results to the SAFT domain.

 In the next section, we will introduce some of these classical results that we will extend to the SAFT domain.

\section{Preliminaries}
 Shift-invariant spaces have been the focus of many research papers in recent years because of their close connection with sampling theory \cite{unser2000,eldar2009} and wavelets and multiresolution analysis \cite{Walter1992,Janssen1993,WalterZ}. They have many applications in signal and image processing \cite{Bhandari2012}. For example, in
many signal processing applications, it is of interest to represent
a signal as a linear combination of shifted versions of some basis
function $\varphi ,$ called the generators of the space, that generates a stable basis for a space.
More precisely, we consider spaces of the form
\begin{equation}
\mathcal{V}\left( \varphi  \right) = \left\{ {f\left( t \right) = \sum\limits_{n =  - \infty }^{ + \infty } {c\left[ n \right]\varphi \left( {t - n} \right)} ,\varphi  \in L^2 \left( \mathbb{R} \right),\left\{ {c\left[ n \right]} \right\} \in \ell _2 } \right\}.
\label{sis}
\end{equation}
The closure of $\mathcal{V}(\varphi )$ is a closed subspace of $L^2$. Furthermore,
 it is shift-invariant in the sense that for all $ f \in \mathcal{V}(\varphi )$,
 its shifted version, $f(\cdot - k)\in \mathcal{V}(\varphi )$, $k \in \mathbb{Z},$ where $\Z$ denotes
 the set of integers.

For the basis functions to be stable, it is required that the
family of functions $ \left\{ {\varphi\left( {t - n} \right)}
\right\}_{n =  - \infty }^\infty$ forms a Riesz basis or
equivalently, there exists two positive constants $0< \eta _1,\eta_2 < +\infty ,$
  such that
\begin{equation}
    \forall c \in \ell _2 , \quad \eta _1 \left\| c \right\|_{\ell _2 }^2
\leqslant \left\| {\sum\limits_{n =  - \infty }^\infty
{c[k]\varphi \left( {t - k} \right)} } \right\|_{L^2 }^2
\leqslant \eta _2 \left\| c \right\|_{\ell _2 }^2
    \label{riesz}
\end{equation}
where $\ell_2$ is the space of all square-summable sequences and  $\left\| c \right\|_{\ell _2 }^2$ is the squared $\ell _2$-norm of
the sequence. We define the Fourier transform of $h(t)$ by $\widehat h\left( \omega  \right) = \frac{1}{\sqrt{2\pi}}\int_{ - \infty }^{ + \infty }
{h(t)e^{-j\omega t} dt}$. Recall, the Fourier domain equivalent of (\ref{riesz}) is
\begin{equation}
\eta _1  \leqslant \sum\nolimits_{n =  - \infty }^{ + \infty } {\left|
{\widehat\varphi \left( {\omega  + 2\pi n} \right)} \right|^2 }  \leqslant \eta _2.
    \label{rieszft}
\end{equation}
The ratio $\rho  = {{\eta _2 }}/{{\eta _1 }}$ is called the
\textit{condition number} of the Riesz basis. The basis is
shift-orthonormal \footnote{ Shift-orthonormality means that $
\left\langle {\varphi ,\varphi \left( { \cdot  - k} \right)}
\right\rangle  = \delta _k$ where $\left\langle {x,y}
\right\rangle  = \int_{ - \infty }^{ + \infty } {x(t)y^ *  (t)dt}$
is the $L^2$-inner product and $\delta_{k} = \left\{
\begin{matrix}  1, & \mbox{if } k=0  \\  0, & \mbox{if } k \ne 0
\end{matrix}\right.$ denotes the Kronecker delta. } or a tight
frame if $\rho = 1$.

One of the important tools  used in the study of sampling spaces is the Zak transform \cite{Daubechies:92, Janssen1993}.
 The Zak transform, which was introduced in quantum mechanics by J. Zak \cite{Zak1967} to solve Schr\"{o}dinger's equation for an electron subject to a periodic potential in a constant magnetic field, may be defined as
$$Z_f(t,\omega)=\sum\nolimits_{k=-\infty}^{+\infty}  f(t+k)e^{-2\pi jk \omega},  f\in L^1(\mathbb{R})$$
  It is easy to see that
 $$Z_f(t,\omega +1)=Z_f(t,\omega) \mbox{ and }  Z_f(t+1,\omega)=e^{2\pi j\omega}Z_f(t,\omega)$$
  and that the Zak transform is a unitary transformation from $L^2(\mathbb{R})$ onto $L^2(Q),$ with $\left\|Z_f\right\|_{L^2(Q)}=\left\|f\right\|_{L^2(\mathbb{R})},$  where
  $Q$ is the unit square $Q=[0,1]\times [0,1] .$

  To see the connection between the Zak transform and sampling spaces, let us for the sake of convenience define the Fourier transform
  of $f$ as
  $$\hat{f}(w)=\int_{-\infty}^\infty f(t)e^{-2\pi j wt}dt$$
  so that the inverse transform, whenever it exists, is given by
  $$f(t) =\int_{-\infty}^\infty \hat{f}(w) e^{2\pi j wt}dw .$$
  Now if $f$ belongs to a sampling space with sampling function $\psi \in L^2\cap L^1,$
  it is easy to see
  that since $f\left( t \right) = \sum\nolimits_{k =  - \infty }^{ + \infty } {f\left( k \right)\psi \left( {t - k} \right)}, $ then
  $$\hat{f}(w)=\hat{F}(w)\hat{\psi}(w),\ \ \  \mbox{ where } \ \ \ \hat{F}(w)=\sum\nolimits_{k =  - \infty }^{+\infty}f(k)e^{-2\pi \jmath kw}.$$
  Clearly $\hat{F}(w)$ is periodic with period one; hence,
  $$\left| \hat{f}(w+k)\right|=\left| \hat{F}(w)\right|\left| \hat{\psi}(w+k)\right|,$$
  which, in view of the fact that $\hat{F}(w)=Z_f(0, w),$ implies that
  $$G_f(w)=\left|Z_f(0, w)\right|^2 G_{\psi}(w),$$
  where $G_g$ denotes the Grammian of $g\in L^2(\mathbb{R})$,
defined by $G_g(w)=\sum_{k\in
\mathbb{Z}}\left|\hat{g}(w+k)\right|^2. $
  Therefore, for such a function $f$, we have
  \begin{equation}
  A \left|Z_f(0, w)\right|^2\leq G_f(w)\leq B \left|Z_f(0, w)\right|^2,\label{Z1}
  \end{equation}
  for some $A,B>0 .$ From this, we obtain
  \begin{equation}
A\sum\nolimits_{k =  - \infty }^{+\infty} \left\vert f({k})\right\vert ^{2}\leq
\int_{0}^{1}G_{f}(w)dw\leq B \sum\nolimits_{k =  - \infty }^{+\infty} \left\vert
{f}({k})\right\vert ^{2}, \label{Z2}
\end{equation}
and \begin{equation}
 \int_0^1\frac{ \sum_{k\in \mathbb{Z}}\left|\hat{f}(w+k)\right| }{\left|Z_f(0,w)\right|}dw=\int_{\mathbb{R}} \left|\hat{\psi}
 (w)\right| dw <\infty , \label{Z3}
 \end{equation}
 whenever $\hat{\psi} \in L^1(\R).$  The Shannon sampling theorem is also known to follow from the Poisson summation formula
  \begin{equation}\sum\nolimits_{k =  - \infty }^{+\infty}f(t+ k)=\sum\nolimits_{k =  - \infty }^{+\infty}\hat{f}(k)e^{2\pi jkt},\label{Pois1}
  \end{equation}
  or equivalently
\begin{equation}\sum\nolimits_{k =  - \infty }^{+\infty}\hat{f}(w+ k)=\sum\nolimits_{k =  - \infty }^{+\infty}{f}(k)e^{-2\pi jkw},\label{Pois2}
  \end{equation}
  which, when $f$ is band-limited to $(-1/2,1/2),$ leads to
  \begin{equation}\sum\nolimits_{k =  - \infty }^{+\infty}{f}(k)=\int_{\mathbb{R}}{f}(t)dt . \label{Pois3}
  \end{equation}
Moreover, we have in view of (\ref{Pois2}), that
$Z_f(0,w)=\sum\nolimits_{k =  - \infty }^{ + \infty } { \hat{f}(w+ k)}.$

\section{Convolution Structures}
In this section we introduce convolution operations associated with the SAFT, one for functions and one for a sequence of numbers and a function. Furthermore, we introduce a definition for the discrete SAFT that will be useful for deriving Poisson summation formula and Zak transform for SAFT.

In order define convolution operators associated with the SAFT
, let us first define unitary, modulation operation.
 Let $\Lu{t} = \exp \left( {\jmath \frac{{a{t^2}}}{{2b}}} \right)$ be the chirp--modulation function and let us define,
	\begin{align*}
	\up{f}{t} & \DE \Lu{t} f\left(t\right) = {e^{\jmath \frac{{a{t^2}}}{{2b}}}}f\left( t \right) \\
	\dn{f}{t} & \DE \Ld{t} f\left(t\right) = {e^{-\jmath \frac{{a{t^2}}}{{2b}}}}f\left( t \right)
	\end{align*}
 where $\overline{z}$ means the conjugate of $z.$
\begin{definition}[SAFT--Convolution]
We define the convolution $\SC$ of two functions $f$ and $g$ as
\begin{equation}
\label{saftconv}
\left( {f \SC g} \right)\left( t \right) = \frac{\Ld{t}}{{\sqrt {\left| b \right|} }}\left( {\up{f}{t} * \up{g}{t}} \right)\left( t \right)
\end{equation}
where $*$ stands for the standard convolution, that is,
$$ \left(f* g\right)(t)=\frac{1}{\sqrt{2\pi}}\int_{\R} f\left(t-x\right)g(x)dx.$$
\end{definition}
We will show that the convolution operator in (\ref{saftconv}) leads to the well-known duality property of the Fourier transform.  A similar definition was presented in \cite{Xiang2012} but our proof relies on \cite{BZ:2015, Bhandari2012,Zayed1996}. 
\begin{theorem}
Let $h(t)=(f \SC g)(t).$ Then the SAFT $H(\omega)$ of $h$ is given by
$$ H(\omega)= \overline{\eta}_\mathbf{A}(\omega) F(\omega) G(\omega),$$
where
${\eta}_\mathbf{A} \left( \omega  \right) = \exp \left( {\frac{\jmath }{{2b}}\left( {d{\omega ^2} + \Omega \omega } \right)} \right)$ and $\Omega =2(bq-dp).$
\end{theorem}
\begin{proof}
{\small
\begin{eqnarray*}
H(\omega)&=& \int_{\R} h(t) k(t,\omega) dt \\
&=& \int_{\R} \underbrace{\left( \frac{\Ld{t}} {\sqrt{2\pi|b|}}\int_{\R} \up{f}{t-x} \up{g}{x} dx \right)}_{h\left( t \right)} k\left(t,\omega \right) dt \\
&=& \int_{\R} k(t,\omega)\frac{\Ld{t}}{\sqrt{2\pi|b|}}\int_{\R}  {f}(t-x)\exp\left( \frac{\jmath}{2b}\left( a(t-x)^2\right)\right)
   {g}(x)\exp\left( \frac{\jmath}{2b}\left( ax^2\right)\right) dx dt\\
   &=& \Nd{\omega} \frac{1}{{2\pi |b|}}\int_{\R}f(t-x)\exp\left( \frac{\j}{2b}\left( a(t-x)^2+d\omega^2-2\omega(t-x) +\Omega \omega+2p(t-x) \right)\right) dt\\
   &\times & \int_{\R}g(x)\exp\left( \frac{\j}{2b}\left( ax^2+d\omega^2-2\omega x +\Omega \omega+2px \right)\right ) dx\\
  &=& \Nd{\omega} \frac{1}{{2\pi |b|}}\int_{\R}f(u)\exp\left( \frac{\j}{2b}\left( au^2+d\omega^2-2\omega u +\Omega \omega+2pu \right)\right)du\\
   &\times & \int_{\R}g(x)\exp\left( \frac{\j}{2b}\left( ax^2+d\omega^2-2\omega x +\Omega \omega+2px \right)\right)dx\\
    &=& \Nd{\omega} F(\omega)G(\omega).
\end{eqnarray*}}
\end{proof}

\noindent The following definition of the convolution of a sequence and a function is a generalization of that given in \cite{Bhandari2012}.
\begin{definition}[Discrete Time SAFT (DT--SAFT)]
\label{DiscSFAT}
Let $P=\left\{ p(k)\right\}$ be a sequence in $\ell^2,$ that is, $\sum_k \left|p(k)\right|^2<\infty.$  We define the discrete time SAFT of $P$ as
\begin{equation}
\dtsaft{P}=\frac{1}{\sqrt{2\pi |b|}} \sum_k p(k)\exp\left\{ \frac{j}{2b}\left( ak^2+d\omega^2-2\omega k +\Omega \omega+2pk \right)\right\},
 \end{equation}
and define the convolution of a sequence $P$ and a function $\phi \in L^2(\R)$ as
$$ h(t)=(P \SC \phi)(t)= \frac{1}{\sqrt{2\pi |b|}} e^{-jat^2/2b} \sum_k e^{jak^2/2b} p(k)e^{ja(t-k)^2/b} {\phi}(t-k) . $$
\end{definition}
\begin{lemma}
Let $P$ and $\phi$ be as above and $h(t)=(P*_{\bsls} \phi)(t)$. Then
$$H(\omega)=\dtsaft{h}=\overline{\eta}_{\bsls}(\omega) \dtsaft{P} \Phi_{\bsls}(\omega),$$
where $\Phi_{\bsls}$ is the SAFT of $\phi.$
Moreover, $\left| \dtsaft{P} \right|$ is periodic with period $\Delta=2\pi b.$ \label{lemma1}
\end{lemma}
\begin{proof}
From the definition of SAFT, we have
\begin{eqnarray*}
H(\omega) &=& \int_{\R} h(t)k(t,\omega) dt= \frac{1}{{2\pi |b|}} \sum_k e^{jak^2/2b} p(k)\int_{\R} e^{-jat^2/2b}\phi (t-k)e^{ja(t-k)^2/2b}\\
&\times & \exp\left\{ \frac{j}{2b}\left( at^2+d\omega^2-2\omega t +\Omega \omega+2pt \right)\right\}dt\\
&=&\frac{1}{{2\pi |b|}} \sum_k e^{jak^2/2b} p(k)\int_{\R} \phi (u)
\exp\left\{ \frac{j}{2b}\left( au^2+d\omega^2-2\omega (u+k) +\Omega \omega+2p(u+k) \right)\right\}du\\
&=& \frac{\overline{\eta}_{\bsls}(\omega)}{{2\pi |b|}} \sum_k \exp\left\{\frac{j}{2b}\left( ak^2-2k\omega+2pk\right) \right\} p(k) \times  \int_{\R} \phi (u)
\exp\left\{ \frac{j}{2b}\left( au^2+d\omega^2-2\omega u +\Omega \omega+2p u \right)\right\}du \\
&=& \frac{\overline{\eta}_{\bsls}(\omega)}{\sqrt{2\pi |b|}}\left[ \sum_k \exp\left\{\frac{j}{2b}\left( ak^2+d\omega^2-2k\omega+\Omega \omega + 2pk\right) \right\} p(k)\right]\Phi_{\bsls}(\omega)\\
&=& \overline{\eta}_{\bsls}(\omega)\hat{P}_{\bsls}(\omega)\Phi_{\bsls}(\omega).
\end{eqnarray*}
Furthermore, since $e^{-jk\Delta/b}=e^{-2jk\pi}=1,$ we have
\begin{eqnarray*}
\hat{P}_{\bsls}(\omega +\Delta)(\omega)&=& \frac{1}{\sqrt{2\pi |b|}} \sum_k p(k)\exp\left\{ \frac{j}{2b}\left( ak^2+d(\omega+\Delta)^2-2k(\omega +\Delta)
  +\Omega (\omega + \Delta) +2pk \right)\right\}\\
  &=& \frac{1}{\sqrt{2\pi |b|}} \sum_k p(k)\exp\left\{ \frac{j}{2b}\left( ak^2+d\omega^2-2k\omega
  +\Omega \omega +2pk \right)\right\}\\
  &\times & \exp\left\{ \frac{j}{2b}\left( d\Delta^2+ 2d\omega\Delta-2k\Delta
  +\Omega \Delta \right)\right\}\\
  &=& \hat{P}_{\bsls} (\omega)\exp\left\{ \frac{j}{2b}\left( d\Delta^2+ 2d(\omega)\Delta
  +\Omega \Delta \right)\right\}
\end{eqnarray*}
Thus,
$$\left| \hat{P}_{\bsls}(\omega +\Delta)\right|= \left| \hat{P}_{\bsls}(\omega)\right|.$$
\end{proof}
In the next theorem we give a necessary and sufficient condition for a function $\phi(t)$ to be a generator for a shift-invariant space in terms of its SAFT.

\begin{theorem}
\label{Thm:Riesz}
 Let $P=\left\{  p(n)\right\}  \in\ell_{2},$ $\phi \in L^{2}(\mathbb{R})$
and
consider the chirp-modulated shift-invariant subspaces of $L^{2}(\R)$%
$$ \mathcal  V\left( \phi  \right) = {\mathrm{closure}}\Bigg\{ f \in L^2: f(t) = \left( P\star_{A}
\phi \right) (t)\Bigg\}.$$
Then $\left\{ e^{j(t-k)^2/2b}\phi(t-k)\right\}  $ is a Riesz basis for
$\mathcal{V}(\phi)$ if and only if there exist two positive constants $\eta
_{1},\eta_{2}>0$ such that%
\begin{equation}
\label{frame}
\eta_{1}\leq\sum\limits_{k =  - \infty }^{+\infty}\left| {\Phi}_{A}(w+k)\right| ^{2}\leq\eta_{2}
\end{equation}
for all $w\in\left[  0,\Delta\right] ,$ and ${\Phi}_{A}$ is SAFT of $\phi .$
\end{theorem}
\begin{proof}
Since $f(t)=\left( P\star_{A}\phi \right)(t),$ we have by the previous lemma, $${F}_{A} (w)=\overline{\lambda}_{A}(w){\hat{P}}_{A} (w){\Phi}_{A} (w);$$ and hence $$\left|{F}_{A} (w)\right|^2=
\left|\hat{P}_{A} (w) \right|^2 \left| {\Phi}_{A} (w)\right|^2,$$ where $\hat{P}_{A} (w)$ is given by Definition (\ref{DiscSFAT}).
Thus, because $\left|\hat{P}_{\bsls} (\omega)\right|$ is periodic with period $\Delta$
\begin{align*}
  \left\|  F_{\bsls}  \left( \omega  \right) \right\|_{L^2 \left( R \right)}^2 & =
   \int\limits_{ - \infty }^{  \infty } \left| \hat{P}_{\bsls}  \left( \omega  \right)\right|^2 \left|\Phi _{\bsls}  \left( \omega  \right) \right|^2 d\omega   \hfill \\
  & = \sum\limits_{k =  - \infty }^{  \infty } \int\limits_{k\Delta }^{(k + 1)\Delta } \left| \hat{P}_{\bsls}  \left( \omega  \right)\right|^2 \left|\Phi _{\bsls} \left( \omega  \right) \right|^2 d\omega   \hfill \\
  & = \sum\limits_{k=  - \infty }^{  \infty } \int_0^\Delta  \left| \hat{P}_{\bsls} \left( \omega  + k\Delta  \right) \right|^2 \left| \Phi _{\bsls} \left( \omega  + k\Delta  \right) \right|^2 d\omega  .  \hfill \\
  & =  \int_0^\Delta  \left|  \hat{P}_{\bsls} \left( \omega   \right) \right|^2 \sum\limits_{k=  - \infty }^{  \infty }\left| \Phi _{\bsls} \left( \omega  + k\Delta  \right) \right|^2 d\omega    \hfill \\
  &= \int_0^\Delta \left|  \hat{P}_{\bsls} \left( \omega   \right) \right|^2 G_{\phi,A}(\omega) d \omega ,
\end{align*}
where $G_{\phi,A}(\omega)=\sum_k \left| \Phi_{\bsls} (\omega + k \Delta)\right|^2 $ is the Grammian of $\phi.$
But
$$
\int_0^\Delta \left|\hat{P}_{\bsls}(\omega)\right|^2 d\omega = \frac{1}{2\pi b}\int_0^\Delta \sum_{k,l}p(k)\overline{p}(l)\exp\left\{\frac{j}{2b}\left[
a(k^2-l^2)-2\omega (k-l)+2m(k-l)\right]\right\} d\omega
$$
and since
$$\int_0^\Delta e^{{j}(-2\omega)(k-l)/2b}d\omega=\int_0^\Delta e^{{j}(\omega)(l-k/b)}d\omega =b\int_0^{2\pi} e^{ju(l-k)}du=2\pi b\delta_{k,l} ,$$
it follows that
$$ \left\|\hat{P}_{\bsls}\right\|_{L^2[0,\Delta]}^2=\int_0^\Delta\left|\hat{P}_{\bsls}(\omega)\right|^2d\omega=\sum_k \left|p(k)\right|^2=\left\|p(k)\right\|_{\ell^2}^2.   $$
Since $$0< \eta_1\leq G_{\phi,A}(\omega)\leq \eta_2<\infty $$ and $$ \left\|p(k)\right\|_{\ell^2}^2=\left\|e^{{j}(ak^2)/{2b}} p(k)\right\|_{\ell^2}^2 $$
we have
$$\eta_1 \left\|\hat{P}_{\bsls}\right\|^2=\eta_1 \left\|p(k)\right\|^2\leq \left\|F_{\bsls}\right\|^2\leq \eta_2 \left\|p(k)\right\|^2 \leq \eta_2 \left\|\hat{P}_{\bsls}\right\|^2$$
which completes the proof.\end{proof}

To get orthonormal basis for the $\mathcal  V\left( \phi  \right)$ we use the standard trick of putting
$$H_{\bsls}(\omega)=\frac{\Phi_{\bsls} (\omega)}{\sqrt{ G_{\phi, A}(\omega)}}$$
so that
$$ \sum_k \left|H_{\bsls}(\omega +k\Delta)\right|^2=\frac{1}{G_{\phi, A}(\omega)}\sum_k \left|\Phi_{\bsls}(\omega +k\Delta)\right|^2=1.$$

\section{The Zak Transform Associated with the SAFT}
\begin{definition} \label{defzak}
We define the Zak transform  associated with the SAFT of a signal $f$ as
$$ Z_{\bsls}(t,\omega)=\frac{1}{\sqrt{2\pi b}}\sum_k f(t+k)\exp\left\{ \frac{j}{2b}\left[ d\omega^2 + ak^2 -2k\omega +\Omega \omega +2pk\right]\right\}$$
\end{definition}
We have the following theorem
\begin{theorem} The Zak transform given by definition \ref{defzak} is an isometry between $L^2(\R)$ and $L^2(B),$ where $\mbox{B}=[0,1]\times[0, \Delta],$ that is there is a one-to-one correspondence between $f\in L^2(\R)$ and $Z_A \in L^2(B)$ such that $\left\|f\right\|_{L^2(\R)}^2 = \left\|Z_{\bsls}\right\|_{L^2(\mbox{B})}^2.$
\end{theorem}
 \begin{proof}
 First, let us observe that since $$e^{-jk\Delta/b}=e^{-2jk\pi }=1,$$ it follows that
  \begin{align*}
  Z_{\bsls}(t,\omega +\Delta)& =  \frac{1}{\sqrt{2\pi b}}\sum_k f(t+k)\\
  &\times \exp\left\{ \frac{j}{2b}\left[ d(\omega+\Delta)^2 +ak^2 -2k(\omega+\Delta) +\Omega (\omega +\Delta) +2pk\right]\right\}\\
  &= \exp\left\{ \frac{j\Delta}{2b}\left[ d\Delta +2d \omega+ +\Omega\right]\right\}Z_{\bsls}(t,\omega)
  \end{align*}
Thus $ \left|Z_{\bsls}( t,\omega +\Delta)\right|^2=  \left|Z_{\bsls}( t,\omega )\right|^2$, and we have
 \begin{eqnarray*}
  \int_0^\Delta \left|Z_{\bsls}( t,\omega )\right|^2 d\omega & = & \frac{1}{{2\pi b}}\sum_{k,l} f(t+k)\overline{f}(t+l)\\
 &\times&  \int_0^\Delta \exp\left\{ \frac{j}{2b}\left[ a(k-l)^2 -2\omega(k-l) +2m(k-l)\right]\right\}d\omega \\
  &=& \frac{1}{{2\pi b}}\sum_{k,l} f(t+k)\overline{f}(t+l)
   \exp\left\{ \frac{j}{2b}\left[ a(k-l)^2  +2m(k-l)\right]\right\}\\
   &\times & \int_0^\Delta e^{-j\omega(k-l)/b} d\omega\\
   &=& \frac{1}{{2\pi b}}\sum_{k,l} f(t+k)\overline{f}(t+l)
   \exp\left\{ \frac{j}{2b}\left[ a(k-l)^2  +2m(k-l)\right]\right\}\\
   &\times & \int_0^{2\pi} e^{-jx(k-l)} b dx\\
   &=& \sum_{k} \left|f(t+k)\right|^2.
  \end{eqnarray*}
Therefore
\begin{align*}
\left\|Z_{\bsls}\right\|_{L^2(\mbox{B})}^2& =\int_0^1\int_0^\Delta \left|Z_{\bsls}( t,\omega )\right|^2 d\omega dt = \int_0^1  \sum_{k} \left|f(t+k)\right|^2 dt\\
& = \int_{-\infty}^\infty \left|f(t)\right|^2 dt =\left\|f\right\|_{L^2(\R)}^2.
\end{align*}
\end{proof}
Now we derive the analogue of Formulae (\ref{Z3}) . It is easy to see that
$$Z_{A,f}(0, \omega)=\hat{F}_{\bsls}\left( \omega, \left\{k\right\}\right) ,$$ where $\hat{F}_{\bsls}\left( \omega, \left\{k\right\}\right)$ is the SAFT of the sequence of samples  $\left\{f(k)\right\}$ of $f.$    Let $f$ be in the sampling space generated by $\phi(t)$ ,i.e.,
$$ f(t)=\frac{1}{{2\pi b}}e^{-jat^2/2b}\sum_{k} e^{jak^2/2b}f(k) e^{ja(t-k)^2/2b} {\phi}(t-k). $$
By Lemma \ref{lemma1}  we have
$$ \hat{F}_{\bsls} (\omega)=\overline{\eta}(\omega)\hat{F}_{\bsls}\left( \omega, \left\{k\right\}\right) \Phi_{\bsls}(\omega)=\overline{\eta}(\omega) Z_{A,f}(0, \omega)  \Phi_{\bsls}(\omega)$$
Hence,
$$ \left|\hat{F}_{\bsls} (\omega)\right|^2=\left|  Z_{A,f}(0, \omega)\right|^2 \left|\Phi_{\bsls}(\omega)\right|^2 $$ and because the modulus of the Zak transform is periodic with period $\Delta,$ we have
$$ \left|\hat{F}_{\bsls} (\omega+k\Delta)\right|^2=\left|  Z_{A,f}(0, \omega)\right|^2 \left|\Phi_{\bsls}(\omega+k\Delta)\right|^2  $$
we have
$$ G_{A, f}(\omega)= \left|  Z_{A,f}(0, \omega)\right|^2  G_{A, \phi}(\omega),$$ where $G_{A,f}$ is the Grammian of $f.$
Since
$$\int_0^\Delta {G}_{A, \phi}(\omega)d\omega= \sum\limits_{k =  - \infty }^{+\infty}\int_0^\Delta \left|\hat{\Phi}_{A}(\omega +k \Delta )\right|^{2}=\left\| \Phi_{\bsls}
\right\|^2<\infty,$$
it follows that $ G_{A, \phi}(\omega)<\infty$ almost everywhere. By relation (\ref{frame}), $ \eta_1\leq G_{A, \phi}(\omega)\leq \eta_2,$ and from this we obtain
$$ \eta_1 \leq \frac{G_{A, f}(\omega)}{\left|  Z_{A,f}(0, \omega)\right|^2}\leq \eta_2. $$ Thus,
$$\eta_1 \left|  Z_{A,f}(0, \omega)\right|^2\leq G_{A, f}(\omega) \leq \eta_2  \left|  Z_{A,f}(0, \omega)\right|^2 . $$ But it is easy to see that
\begin{eqnarray*}
\left|  Z_{A,f}(0, \omega)\right|^2 &=& Z_{A,f}(0, \omega)\overline{Z}_{A,f}(0, \omega)\\
&=&\sum_{k,l} f(k)\overline{f}(l)\exp\left\{\frac{j}{2b}\left[a(k^2-l^2)-2\omega (k-l)+2p(k-l)  \right]   \right\}.
\end{eqnarray*}
Therefore, since $\int_0^\Delta \exp\left( -j\omega(k-l)/b\right)d\omega =2\pi b \delta_{k,l}, $ we have
$$\int_0^\Delta  \left|  Z_{A,f}(0, \omega)\right|^2 d\omega= \sum_k \left|f(k)\right|^2, $$ and it follows that
$$ \eta_1 \sum_k \left|f(k)\right|^2 \leq \int_0^\Delta G_{A,f}(\omega) d\omega \leq \eta_2 \sum_k \left|f(k)\right|^2 . $$
Moreover, $${\small
 \int_0^\Delta \frac{\sum_k \left|F_{\bsls}(\omega + k\Delta) \right|}{\left|  Z_{A,f}(0, \omega)\right|}d\omega =\int_0^\Delta \sum_k \left|\Phi_{\bsls}(\omega + k\Delta) \right| d\omega =\int_{\R} \left|F_{\bsls}(\omega)\right|d\omega <\infty}.$$

\section{Poisson Summation Formula for SAFT}
In the next theorem we derive the Poisson summation formula for the SAFT.
\begin{theorem}
 The Poisson summation formula for the SAFT is given by
\begin{align*}
& {\sqrt{2\pi b}}\sum\limits_{k=-\infty}^{+\infty}f(t+k\Delta)\exp\left\{\frac{j}{2b}\left[ a(t+k\Delta)^2+2p(t+k\Delta)\right]\right\} \\
&= \sum\limits_{n =  - \infty }^{ + \infty }\exp\left\{\frac{-j}{2b}\left[ dn^2 +\Omega n-2nt\right]\right\}F_{\bsls}(n).
\label{PSFSAFT}
\end{align*}
\end{theorem}

\begin{proof}
Let $f\left(t\right)$ be an integrable function and define
$$
\tilde{f}\left(t\right)=\frac{1}{\sqrt{2\pi b}}\sum\limits_{k=-\infty}^{+\infty}f(t+k\Delta)\exp\left\{\frac{j}{2b}\left[ a(t+k\Delta)^2+2p(t+k\Delta)\right]\right\}.
$$
It is easy to see that $\tilde{f}\left(t\right)$ is a $\Delta$--periodic function. Therefore,
 it can be expanded as a Fourier series in terms of the orthogonal family $\left\{ e^{j n t/b}\right\}.$ Hence,
 {\small
\begin{eqnarray*}
\tilde{f}\left(t\right)& =& \frac{1}{2\pi b}\sum\limits_{n =  - \infty }^{ + \infty }e^{j n t/b}\int_0^\Delta e^{-j n  y/b}\tilde{f}(y)dy\\
&=& \frac{1}{\sqrt{2\pi b}}\sum\limits_{n =  - \infty }^{ + \infty }e^{j n  t/b}\sum\limits_{k =  - \infty }^{  \infty }\int_0^\Delta e^{- j n  y/b}f(y+k\Delta) \times \exp\left\{\frac{j}{2b}\left[ a(y+k\Delta)^2+2p(y+k\Delta)\right]\right\}dy\\
&=&\frac{1}{\sqrt{2\pi b}}\sum\limits_{n =  - \infty }^{ + \infty }e^{j n  t/b}\sum\limits_{k =  - \infty }^{  \infty }\int_{0}^{\Delta}
f(y+k\Delta) \times \exp\left\{\frac{j}{2b}\left[ a(y+k\Delta)^2+2p(y+k\Delta)\right]-2ny\right\}dy \\
&=& \frac{1}{\sqrt{2\pi b}}\sum\limits_{n =  - \infty }^{  \infty }e^{j n  t/b}\sum\limits_{k = - \infty }^{  \infty }\int_{k\Delta}^{(k+1)\Delta}
f(x) \times \exp\left\{\frac{j}{2b}\left[ ax^2+2mx-2n(x-k\Delta)\right]\right\}dx \\
&=&\frac{1}{\sqrt{2\pi b}}\sum\limits_{n =  - \infty }^{ \infty }\exp\left\{\frac{-j}{2b}\left[ dn^2 +\Omega n-2nt\right]\right\} \times  \int_{-\infty}^\infty f(x)\exp\left\{\frac{j}{2b}\left[ ax^2+ dn^2-2nx +\Omega n+ 2px\right]\right\}dx\\
&=& \sum\limits_{n =  - \infty }^{ + \infty }\exp\left\{\frac{-j}{2b}\left[ dn^2 +\Omega n-2nt\right]\right\}F_{\bsls}(n).
\end{eqnarray*}}
That is,
{\small
\begin{align*}
& {\sqrt{2\pi b}}\sum\limits_{k=-\infty}^{+\infty}f(t+k\Delta)\exp\left\{\frac{j}{2b}\left[ a(t+k\Delta)^2+2p(t+k\Delta)\right]\right\}\\
& =\sum\limits_{n =  - \infty }^{ + \infty }\exp\left\{\frac{-j}{2b}\left[ dn^2 +\Omega n-2nt\right]\right\}F_{\bsls}(n).
\end{align*}}
\end{proof}

In particular, if $F_{\bsls}$ is an interpolating function, i.e., $F_{\bsls}(n)=\delta (n)$, then
{\small
$$\sum\nolimits_{k=-\infty}^{+\infty}f(t+k\Delta)\exp\left\{\frac{j}{2b}\left[ a(t+k\Delta)^2+2m(t+k\Delta)\right]\right\}=\frac{1}{\sqrt{2\pi b}}.$$}

\begin{lemma}
If $g$ is bandlimited to $(-1,1)$ in the SAFT domain, then
$$ \sum\nolimits_k \left|g(t_k)\right|^2 =\frac{1}{2\pi b} \int_{\R} \left| G_{\bsls}(\omega)\right|^2 d\omega .$$
\end{lemma}
\begin{proof}
Let $f$ be bandlimited to $(-1,1)$ in the SAFT domain.  The Poisson summation formula for $t=0$ and $t_k=k\Delta$ becomes
\begin{equation}
 {\sqrt{2\pi b}}\sum_k f(t_k)\zeta(t_k)=F_{\bsls} (0), \label{PSF2}
 \end{equation}
where $\zeta (t)= \exp\left\{ \frac{j}{2b}\left( at^2+2pt\right) \right\}.$ Setting $f(t)=\overline{\zeta} (t) g(t)\overline{h}(t)$ in the last equation yields
$$ \sum\nolimits_k g(t_k)\overline{h}(t_k)=\frac{1}{{\sqrt{2\pi b}}}F_{\bsls} (0), $$
but
\begin{eqnarray*}
 F_{\bsls}(\omega)&=& \frac{1}{{\sqrt{2\pi b}}}\int_{\R} \overline{\zeta}(t) g(t)\overline{h}(t)\exp\left\{\frac{j}{2b}\left[ at^2+d\omega^2 +\Omega \omega -2\omega t+2pt\right]\right\}dt\\
 &=& \frac{1}{{\sqrt{2\pi b}}}\int_{\R} g(t)\overline{h}(t)\exp\left\{\frac{j}{2b}\left[ d\omega^2 +\Omega \omega -2\omega t \right]\right\}dt;
 \end{eqnarray*}
 hence
 $$F_{\bsls}(0)=\frac{1}{\sqrt{2\pi b}}\int_{\R} g(t)\overline{h}(t)dt .$$
 Therefore, we have
 $$\sum_k g(t_k)\overline{h}(t_k) =\frac{1}{2\pi b}\int_{\R} g(t)\overline{h}(t)dt,$$
 which leads to
 $$\sum_k \left|g(t_k)\right|^2=\frac{1}{2\pi b}\int_{\R} \left|g(t)\right|^2=\frac{1}{2\pi b}\int_{\R} \left|G_{\bsls}(\omega)\right|^2 .$$
\end{proof}

\section{Shannon's Sampling Theorem and the SAFT: \\
Reinterpretation, Extension and Applications}
In this section, we will provide a formal link between Shannon's sampling theorem and (SAFT--) convolution based least--squares approximation. The key idea is that modulated and shifted versions of the $\sinc$--function form a SAFT--bandlimited subspace.  We then discuss strategies for reconstruction/interpolation of functions using arbitrary basis functions---not necessarily $\sinc$ functions. This serves as an extension of the sampling theorem to generic basis functions. Finally, we conclude this section with an application of fractional delay filtering of SAFT bandlimited functions.

\subsection{Revisiting Shannon's Sampling Theorem for SAFT Domain}

A construction of Shannon's sampling theorem for the SAFT was presented in  \cite{Stern2007,Xiang2012,Xiang2013}. Here we present a different proof of the sampling theorem that uses the theory of reproducing-kernels.

Recall that if $f$ is bandlimited to $[-\sigma, \sigma],$ in the Fourier transform domain, then
$$ f(t)= \sum_n f(t_n)\frac{ \sin\left[\sigma(t-t_n)\right]}{\sigma(t-t_n)}, \quad t_n=n\pi/\sigma$$ if and only if
$$ \hat{f}(\omega)=\sqrt{\pi/(2\sigma^2) }\left(\sum_n f(t_n)e^{jt_n \omega}\right)\chi_{[-\sigma, \sigma]}, $$
where $\chi_{\bsls}$ is the characteristic function of $A.$

Likewise, one can prove that if $f$ is bandlimited to $[-\sigma, \sigma],$ in the SAFT domain, then
$$ \tilde{f}(t)= \sum_n \tilde{f}(t_n)\frac{ \sin\left[\sigma(t-t_n)/b\right]}{\sigma(t-t_n)/b}, \quad t_n=n\pi b /\sigma$$ if and only if
$$ F_{\bsls}(\omega)=\frac{\pi b}{\sigma} F_{\bsls}(\left\{f(t_n)\right\})\chi_{[-\sigma, \sigma]}, $$
where $F_{\bsls}(\left\{f(t_n)\right\}$ is the SAFT of the sequence $\left\{f(t_n)\right\}$ as given by Definition \ref{DiscSFAT} and $\tilde{f}(t)=\overline{\zeta} (t) f(t).$ Although the sampling theorem can be derived either directly or from the Poisson summation formula as
in \cite{Butzer2011}, we will derive it using the theory of reproducing-kernels.

\begin{theorem} \label{samplingthm} If $f$ is bandlimited to $[-\sigma, \sigma]$ in the SAFT domain, then
\begin{equation}
 \tilde{f}(t)=\sum_k \tilde{f}(t_n) \frac{\sin \left[ \frac{\sigma}{b}\left( t-t_n \right)\right]}{\frac{\sigma}{b}\left( t-t_n\right)}, \label{samplingformula}
 \end{equation}
where $$ \tilde{f}=\overline{\zeta} (t) f(t),\quad \mbox{ and } t_n =n\pi b/\sigma ,$$ or
\begin{equation}
f\left( t \right) = {e^{ \jmath \frac{{a{t^2}}}{{2b}}}}\sum_{k\in \Z} {f\left( {kT} \right){e^{-\jmath \frac{{a{{\left( {kT} \right)}^2}}}{{2b}} + \jmath p\frac{{t - kT}}{b}}}\operatorname{sinc} \left( { t/T - k} \right)}
\label{syn}
\end{equation}
where $kT=t_k.$
\end{theorem}
\begin{proof}
It is easy to verify that the functions
$$
\phi_n(t)=\sqrt{\sigma/\pi b} \exp\left\{j\left[at^2+2pt\right]/2b\right\}\frac{\sin\left[ \sigma (t-t_n)/b\right]}{\sigma (t-t_n)/b},
$$
form a complete orthonormal family in $L^2(\R)$ with respect to the inner product
$$\langle \phi_m , \phi_n \rangle  =\int_{\R} \phi_m(t)\overline{\phi}_n (t) dt=\delta_{m,n}, $$
and consequently we have from the theory of reproducing kernels that the reproducing kernel of the space of functions bandlimited to $[-\sigma, \sigma ]$ in SAFT domain is
$$ k(x,t) =\sum\nolimits_n \phi_n(t)\overline{\phi}_n (x) .$$
Hence, one can easily verify that
$$ k(x,t)= e^{\frac{j}{2b}\left[a(x^2-t^2)+2p(x-t) \right]}\frac{\sin \left[ \frac{\sigma}{b}\left( x-t \right)\right]}{\pi \left( x-t\right)}, $$ so that
$$  f(t) =\int_{\R} f(x)k(x,t) dx . $$
Thus,
\begin{eqnarray*}
f(t) &=& \int_{\R} f(x)\sum \phi_n(t) \overline{\phi}_n(x) dx\\
&=& \sum \phi_n(t) \int_{\R} f(x) \overline{\phi}_n(x) dx .
\end{eqnarray*}
But because $\phi_n (t_k)=\sqrt{\sigma/\pi b}\,\zeta(t_k)\delta_{n,k},$ we have
$$ f(t_k)=\sum \phi_n(t_k) \int_{\R} f(x) \overline{\phi}_n(x) dx = \sqrt{\sigma/\pi b}\,\zeta(t_k)\int_{\R} f(x) \overline{\phi}_k(x) dx, $$
and it follows that
$$ f(t)=\sqrt{\pi b/\sigma} \sum f(t_n)\,\overline{\zeta}(t_n)\phi_n(t) ,$$ which is Eq. (\ref{samplingformula}).
\end{proof}

 The sampling theorem can be put in the following form: Let $f\l t\r $ be a continuous time signal such that $| \widehat f \l \omega \r | = 0, | \omega | > \sigma$. Then, $f\l t \r$ is completely determined by its equidistant samples spaced $T = \pi b / \sigma$ seconds apart.
\label{thm:ShannonSAFT}
The reconstruction formula  is then specified by Eq. (\ref{syn}).

\begin{remark}[Generalization of Shannon's Sampling Theory] An immediate consequence of Theorem \ref{samplingthm} is that it applies to all the unitary transforms listed in Table \ref{tab:1}.
\end{remark}
The sampling theorems for the Fresnel Transform \cite{Gori}, fractional Fourier transform \cite{Xia1996, Zayed1996,Bhandari2010, Bhandari2012} and the linear canonical transform \cite{Tao2008, Zhao2009, Shi2012} are all now a straight--forward consequence of Theorem~\ref{samplingthm}. For example, the sampling formula for the LCT can be easily obtained from our sampling theorem by setting $p=0$ to yield
$$f\left( t \right) = {e^{ \jmath \frac{{a{t^2}}}{{2b}}}}\sum_{k\in \Z} {f\left( {kT} \right){e^{-\jmath \frac{{a{{\left( {kT} \right)}^2}}}{{2b}} }}\operatorname{sinc} \left( { t/T - k} \right)}.$$

Unlike previous approaches in \cite{Stern2007,Xiang2012,Xiang2013}, we will show that sampling of SAFT--bandlimited signals amounts to orthogonal projection of the signal onto a subspace of SAFT bandlimited functions. Even more so, the projection amounts to filtering/SAFT--convolution of the signal with low--pass filter followed by the sampling step. The reconstruction process in (\ref{syn}) is simply a semi--discrete SAFT convolution.

Let ${\varphi _0} = \varphi$. Consider basis functions of form,
\begin{equation}
{\varphi _n}\left( t \right) = \frac{1}{\sqrt{T}}\cdot{e^{ - \jmath \frac{{a{t^2} - a{{\left( {nT} \right)}^2}}}{{2b}}}}{e^{ - \jmath p\frac{{t - nT}}{b}}}\operatorname{sinc} \left( {{t}/{T} - n} \right).
\label{basis}
\end{equation}

The family $\{ \varphi_n\}_{n \in\mathbb{Z}}$ has two interesting properties. \\

\noindent\textbf{(P1) Orthonormality:} By construction, the basis functions are orthonormal. This is easy to check. Assuming $T=1$, we have, $\left\langle {{\varphi _n},{\varphi _k}} \right\rangle  = {\delta _{n - k}}$ \cite{unser2000},
$$  \left\langle {{\varphi _n},{\varphi _k}} \right\rangle  = {e^{\jmath \frac{{a{{\left( n \right)}^2} - a{{\left( k \right)}^2}}}{{2b}}}}{e^{\jmath p\frac{{\left( {n - k} \right)}}{b}}}\underbrace {\int {\operatorname{sinc} \left( {t - n} \right)\operatorname{sinc} \left( {t - k} \right)dt} }_{{\delta _{n - k}}}.$$
\noindent\textbf{(P2) Bandlimitedness:} Let $T\sigma = \pi b$. Using (\ref{SAFTD}), we have,
 $$\widehat \varphi\left( \omega  \right) = \frac{{{e^{\frac{\jmath }{{2b}}\left( {d{\omega ^2} + 2\Nu{\omega}} \right)}}}}{{\sqrt {2\pi \left| b \right|} }} \ind_{[-\sigma,\sigma]} \l \omega \r, \quad \sigma  = \frac{{\pi b}}{T}$$
 where $\ind_{\mathrm{ D}} \l \cdot \r$ is the characteristic function of the set $\mathrm{D}$. Since,
$$\left | \widehat \varphi_{\bsls}\left( \omega  \right) \right| =\left |  \frac{\ind_{[-\sigma,\sigma]} \l \omega \r} { {\sqrt {2\pi \left| b \right|} }} \right |, $$
we observe that $\widehat\varphi \l \omega \r$ is a bandlimited function.

Let us define the space of SAFT--bandlimited functions,
\[\mathcal{V}_{\bsl}\l \varphi \r = \left\{ {f\left( t \right) = \sum\limits_{k =  - \infty }^{ \infty } {f\left( t_k \right){\varphi _n}\left( t \right): \left\{f(t_k)\right\} \in {\ell _2}} } \right\}.\] Thanks to \textbf{(P1)}, for any $f \in L_2$, the orthogonal projection operator $P_{\mathcal{V}_{\bsl}\l \varphi \r}: L_2 \rightarrow {\mathcal{V}_{\bsl}\l \varphi \r}$ prescribed by,
\[P_{\mathcal{V}_{\bsl}\l \varphi \r}f = \sum\limits_{k =  - \infty }^{\infty } {c\left[ k \right]{\varphi _n}\left( t \right)} \quad \mbox{with} \quad c = \left\langle {f,{\varphi _n}} \right\rangle \]
results in Shannon's Sampling theorem associated with the SAFT domain. The inner--product of the function $f$ with the basis--functions $\varphi_n$ is equivalent to pre-filtering (in sense of SAFT convolution) with the ideal low--pass filter, followed by uniform sampling. Let us define the low--pass kernel/filter $\psi \left( t \right) = \sqrt {2\pi \left| b \right|} {e^{ - j\frac{{a{t^2}}}{{2b}}}}{e^{ - jp\frac{t}{b}}}{\mathop{\rm sinc}\nolimits} \left( { - t/T} \right)$. Mathematically, we have the following relation,
\[\left\langle {f,{\varphi _n}} \right\rangle  = {\left. {\left( {f{ *_{\bsls} }\psi } \right)\left( t \right)} \right|_{t = kT,k \in \mathbb{Z}}}.\]
 This is because,
\begin{align*}
  \left\langle {f,{\varphi _n}} \right\rangle & = \int {f\left( \tau  \right)\varphi _n^*\left( \tau  \right)d\tau }  \hfill \\
  & = \int {f\left( \tau  \right){e^{j\frac{{a{\tau ^2} - a{{\left( {kT} \right)}^2}}}{{2b}}}}{e^{jp\frac{{\tau  - kT}}{b}}}\operatorname{sinc} \left( {\frac{\tau }{T} - k} \right)d\tau }  \hfill \\
  & = {e^{ - j\frac{{a{{\left( {kT} \right)}^2}}}{{2b}}}}\int {f\left( \tau  \right){e^{j\frac{{a{\tau ^2}}}{{2b}}}}{e^{jp\frac{{\tau  - kT}}{b}}}\operatorname{sinc} \left( {\frac{\tau }{T} - k} \right)d\tau }  \hfill \\
  & = {\left. {\left( {f{*_{\bsls} }\psi } \right)\left( t \right)} \right|_{t = kT,k \in \mathbb{Z}}},
\end{align*}
which leads to Shannon's Sampling theorem for the SAFT domain. More precisely,
\begin{align*}
P_{\mathcal{V}_{\bsl}\l \varphi \r}f  &= \sum\limits_{k =  - \infty }^{\infty } {\left\langle {f,{\varphi _k}} \right\rangle {\varphi _k}\left( t \right)}  \hfill \\
  & = \sum\limits_{k =  - \infty }^{ \infty } {{{\left. {\left( {f{ \SC }\psi } \right)\left( t \right)} \right|}_{t = kT}}{\varphi _k}\left( t \right)}  \hfill \\
  & = {e^{-j\frac{{a{t^2}}}{{2b}}}}\sum\limits_{k =  - \infty }^{ \infty } {f\left( {kT} \right){e^{j\frac{{a{{\left( {kT} \right)}^2}}}{{2b}}}}{e^{ - jp\frac{{t - kT}}{b}}}\operatorname{sinc} \left( {\frac{t}{T} - k} \right)},
\end{align*}
where we have used the property of bandlimited functions,
\[ \widehat f\left( \omega  \right) \in {\ind_{\left[ { - \frac{{\pi b}}{T}, + \frac{{\pi b}}{T}} \right]}}\left( \omega  \right) \Rightarrow \left( {f{ \SC }\psi } \right)\left( t \right) = f\left( t \right).\]
%
%
%
%
%
Hence, Shannon's sampling theorem for SAFT domain can be reinterpreted as orthogonal projection of a signal onto a subspace of SAFT--bandlimited signals. This interpretation provides an approximation  of the sampling theorem that is not discussed in \cite{Stern2007, Xiang2012, Xiang2013}. Due to the involvement of the orthogonal projection operator, the approximation of $f \in L_2$, that is, $\tilde f = P_{\mathcal{V}_{\bsl}\l \varphi \r}f$ is optimal in the least--squares sense since,
\[\widetilde f = \arg \mathop {\min }\limits_x \left\| {f - x} \right\|_{{L_2}}^2.\]
In \cite{Stern2007, Xiang2012, Xiang2013}, the authors discuss the case of interpolation, i.e.,
$c[m] = f\l mT \r, m\in \mathbb{Z}$. Our setup considers a general setting where $f \in L_2$ and $c = \left\langle {f,{\varphi _n}} \right\rangle$.

Also, by invoking the semi--discrete convolution operator, the approximation of a signal $f\in L_2$ is characterised by,
\[
P_{\mathcal{V}_{\bsl}\l \varphi \r}f = {\left. {\left( {f{ \SC }\psi } \right)} \right|_{t = kT}}\SDC \varphi\l t \r \mbox{ \small{(Least--Squares Optimal) }}.
\]

\subsection{Reconstruction with Arbitrary Basis Functions}
In many cases of interest and for practical purposes, one is interested in the expansion of some function $f$ (not necessarily SAFT--bandlimited) in terms of some basis function $\nu$ such that $\psi \left( t \right) = \sqrt {2\pi \left| b \right|} {e^{ - j\frac{{a{t^2}}}{{2b}}}}{e^{ - jp\frac{t}{b}}}\nu \left( t \right)$. In view of Theorem \ref{Thm:Riesz}, we assume that the basis functions generated by $\psi$ form a Riesz basis. More precisely, we consider expansions of form,
\begin{equation}
f\left( t \right) = {e^{ - j\frac{{a{t^2}}}{{2b}}}}\sum\nolimits_k {{p\left[ k \right]}{e^{j\frac{{a{k^2}}}{{2b}}}}{e^{ - jp\frac{{\left( {t - k} \right)}}{b}}}\nu \left( {t - k} \right)}  \Leftrightarrow \left( {p\SDC\psi } \right),
\label{interp}
\end{equation}
which generalizes the representation in (\ref{syn}). Indeed, for the case of sampling series in SAFT domain, we have $p \left[ k \right] = f\l t \r|_{t=k, k\in \Z}$ and $\nu = \sinc$. In order to compute the proxy of Shannon's samples, that is the weights $\{ p \left[ k \right] \}_k$ relative to the basis functions $\nu$, we discretize (\ref{interp}) and obtain,
\begin{equation}
\label{weights}
{f\l k \r} = p\left[ k \right]\SC\left( {\sqrt {2\pi \left| b \right|} {e^{ - j\frac{{a{k^2}}}{{2b}}}}{e^{ - jp\frac{k}{b}}}\nu \left( k \right)} \right) \equiv \left( {p\SC\psi } \right)\left[ k \right].
\end{equation}

Now, assume that there exists a sequence $\vartheta\left[ k \right]$ such that,
\begin{equation}
\label{idf}
p\left[ k \right] = \left( {f\SC\sqrt {2\pi \left| b \right|} {e^{ - j\frac{{a{k^2}}}{{2b}}}}{e^{ - jp\frac{k}{b}}}\vartheta \left[ k \right]} \right).
\end{equation}
Back substituting weights $p\left[ k \right]$ from (\ref{idf}) in (\ref{interp}) results in,
\begin{equation}
\label{AB}
f\left( t \right) = {e^{ - j\frac{{a{t^2}}}{{2b}}}}\sum\nolimits_k {f\left( k \right){e^{j\frac{{a{k^2}}}{{2b}}}}{e^{ - jp\frac{{\left( {t - k} \right)}}{b}}}\underbrace {\left( {\vartheta *\nu } \right)\left( {t - k} \right)}_{{\textsf{New Basis Functions}}}}.
\end{equation}
Hence the representation of $f\l t \r$ in the basis $\nu$ is equivalent to interpolation of samples of $f$ with new basis functions ${\psi _{{\sf{new}}}}\left( t \right) = \sqrt {2\pi \left| b \right|} {e^{ - j\frac{{a{t^2}}}{{2b}}}}{e^{ - jp\frac{t}{b}}}\left( {\vartheta *\nu } \right)\left( t \right)$ and therefore,
\begin{equation}
\label{equivalent}
f\left( t \right) = \left( {p\SDC\psi } \right)\left( t \right) \Leftrightarrow \left( {f\SDC{\psi _{{\sf{new}}}}} \right)\left( t \right).
\end{equation}
The relation between the sequence $\vartheta$ and the basis functions $\nu$ is obtained by sampling (\ref{AB}),
\[
{\left. {f\left( t \right)} \right|_{t = m}} = {e^{ - j\frac{{a{m^2}}}{{2b}}}}\sum\nolimits_k {f\left( k \right){e^{j\frac{{a{k^2}}}{{2b}}}}{e^{ - jp\frac{{\left( {m - k} \right)}}{b}}}\left( {\vartheta *\nu } \right)\left( {m - k} \right)},
\]
which amounts to the interpolation condition, that is,
\begin{align}
\label{IDFcondition}
 \left( {\vartheta *\nu } \right)\left( {m - k} \right) & = {\delta _{m, k}} \quad \mbox{\small (Interpolation Condition)}\\
 \Rightarrow {\left. {f\left( t \right)} \right|_{t = m}} & = {e^{ - j\frac{{a{m^2}}}{{2b}}}}\sum\nolimits_k {f\left( k \right){e^{j\frac{{a{k^2}}}{{2b}}}}{e^{ - jp\frac{{\left( {m - k} \right)}}{b}}}\delta \left( {m - k} \right)}  \notag \\
& = f\left( m \right) \notag
\end{align}
The impulse response of the inverse discrete filter $\vartheta$ is related to the basis functions $\nu$ by a simple Fourier condition,
\begin{equation}
\label{FourierIDF}
\left( {\vartheta *\nu } \right)\left( m \right) \DEq{IDFcondition} {\delta _m} \Leftrightarrow \widehat \vartheta \left( {{e^{j\omega }}} \right)\sum\nolimits_k {\widehat \nu \left( {\omega  + 2\pi k} \right)}  = 1.
\end{equation}
The existence of a stable sequence $\vartheta$ is guaranteed because we assume that $\psi$ (and hence $\nu$) generates Riesz basis and hence for all $\omega \in \mathbb{R}, \ \ \sum\nolimits_k {\widehat \nu \left( {\omega  + 2\pi k} \right)} \neq 0$.

\subsection{Application: Fractional Delay Filtering}
Fractional Delay Filters (FDF) have found applications in a wide range of problems linked with signal processing \cite{Laakso1996,Bhandari2010a}. The FDF problem can be described as follows: Given samples of some finite--energy signal $f\l t \r$, how can one estimate the samples of $\tau$--delayed version of $f$, that is $f^{\tau}\l nT\r \DE f\l t-\tau\r |_{t=nT}, n\in \Z, \tau \in [0,T]$? For the shift--invariant subspace model linked with SAFT domain, we have,
\begin{equation}
\label{FDF}
{f^\tau }\left( t \right) \DEq{interp} \left( {p\SDC\psi } \right)\left( {t - \tau } \right) \DEq{equivalent} \left( {f\SDC{\psi _{{\sf{new}}}}} \right)\left( {t - \tau } \right).
\end{equation}
Since the samples of $f$ are readily available, one option is to set $\nu = \sinc$ which leads to $p\left[ k \right] = f\l kT \r, k\in \Z$. Hence, if $f$ is a SAFT--bandlimited signal, re--samling $\tau$--delayed version of interpolated samples of $f$ leads to desired solution. That said, infinite support and slow--decay of the $\sinc$ filter make this solution impractical for cases when finite number of samples are available or $f$ is not strictly bandlimited. For this purpose, we model $f$ as a shift--invariant signal in SAFT--domain (\ref{interp}). Consequently, we propose to estimate,
\[{f^\tau }\left( {mT} \right) = \left( {p\SDC\psi } \right)\left( {mT - \tau } \right) \approx f\left( {mT - \tau } \right),m \in \Z.\]
This is a two step procedure. Once we set $\nu$ and hence, $\psi$, we compute the inverse discrete filter $\vartheta$. Starting with samples $f\l mT \r$, we compute $p\left[ m \right]$ (\ref{idf}). Then, we compute $\left( {p\SDC\psi } \right)\left( {mT - \tau } \right)$ to estimate delayed signal, $f\left( {mT - \tau } \right)$. The sequence of operations is summarized as follows,
\[f\left( {mT} \right) \to \boxed{{\vartheta }} \DTo{idf} {p} \to \boxed{\psi \left( {mT - \tau } \right)} \DTo{FDF} \left( {p\SDC\psi } \right)\left( {mT - \tau } \right).\]
Next, we present an example of FDF in SAFT Domain.
\subsubsection{Power Cosine Filters for FDF in SAFT Domain} Let $\psi \left( t \right) = \sqrt {2\pi \left| b \right|} {e^{ - j\frac{{a{t^2}}}{{2b}}}}{e^{ - jp\frac{t}{b}}}\nu \left( t \right)$ be as before and,
\begin{equation}
\label{PC}
\nu \left( t \right) = \left( {2/3} \right){\cos ^4}\left( {\pi t/4} \right){\chi _{\left[ { - 2, + 2} \right]}}\left( t \right).
\end{equation}
Let $\widehat \nu\l \omega \r $ be for the Fourier Transform of $\nu\l t \r$. The SAFT of $\psi$ is given by, ${\widehat \psi _{\mathbf{A}}}\left( \omega  \right) = {e^{j\left( {\frac{{d{\omega ^2}}}{{2b}} + \frac{{bq - dp}}{b}\omega } \right)}}\widehat \nu \left( {\omega /b} \right)$. For a unique, stable representation of $f$ in shift--invariant subspace, generated by $\psi$, to exist, the basis functions must form a Riesz basis for SAFT (see Theorem~\ref{Thm:Riesz}). To that end, we must show that the quantity, ${G_{\psi ,{\mathbf{A}}}}\left( \omega  \right) = \sum\nolimits_k {{{\left| {{{\widehat \psi }_{\mathbf{A}}}\left( {\omega  + 2\pi kb} \right)} \right|}^2}}  = \sum\nolimits_k {{{\left| {\widehat \nu \left( {\omega /b + 2\pi k} \right)} \right|}^2}}$ is bounded from below and above. This is indeed the case. Notice that $\widehat \nu \left( \omega  \right) = \sum\nolimits_{\left| k \right| \leqslant 2} {{\rho _k}\operatorname{sinc} \left( {2\omega  - k\pi } \right)}$ with ${\rho _k} = 4/\left( {2 + k} \right)!\left( {2 - k} \right)!$. Let ${G_\nu }\left( \omega  \right) = \sum\nolimits_k {{{\left| {\widehat \nu \left( {\omega  + 2\pi k} \right)} \right|}^2}}$. Since ${G_\nu }\left( \omega  \right)$ is an even symmetric, periodic function, to prove that $G_\nu>0$ (see (\ref{frame}), it suffices to show that $\forall \omega \in [0, \pi]$, $\eta_1 = \inf G_\nu >0$ or equivalently,
\[{\eta _1} = \mathop {\inf }\limits_\omega  \left( {{{\left| {\widehat \nu \left( \omega  \right)} \right|}^2} + {{\left| {\widehat \nu \left( {\omega  - 2\pi } \right)} \right|}^2}} \right),\forall \omega  \in \left[ {0,\pi } \right].\]

Now since $\widehat\nu \l \omega \r$ is a monotonically decreasing function on $[0,\pi]$ and for $\omega = \pi$, $\operatorname{sinc} \left( {2\left( {\omega  + 2\pi m} \right) - k\pi } \right) = 0,\forall m \in \Z - \left\{ 1 \right\}$, we deduce that $\inf_\omega G_\nu$ occurs at $\omega = \pi$. As a result, we have, ${\eta _1} = {G_\nu }\left( \pi  \right) = 2{\left( {\frac{{{{\left( {2!} \right)}^2}}}{4}} \right)^2}$. Using similar argument, we conclude that $\eta_2 = {\sup _\omega }{G_\nu }\left( \omega  \right)$ is computed at $\omega = 0$ which results in $\eta_2 = 1$.

In order to compute the inverse discrete filter $\vartheta$, we use the previously developed property in (\ref{FourierIDF}). With $\nu \left( k \right) = \frac{1}{6}\left( {{\delta _{k - 1}} + 4{\delta _k} + {\delta _{k + 1}}} \right), k\in \Z $, the transfer function of the filter $\vartheta$ results in, $\vartheta \left( {{e^{j\omega }}} \right) = \frac{6}{{{e^{ + j\omega }} + 4 + {e^{ - j\omega }}}}$. Similar to the cubic spline \cite{Bhandari2010a}, the impulse response of such a filter is given by,
\[\vartheta \left[ k \right] =  - \left( {\frac{{6\mu }}{{1 - {\mu ^2}}}} \right){\mu ^{\left| k \right|}},k \in \Z{\text{ with }}\mu  = \sqrt 3  - 2. \]
\subsubsection{Experimental Verification} For experimental verification, we assume uniform samples ${\left\{ {{g_{\sf{sig}}}\left( {kT} \right)} \right\}_k}$ of some signal,
\[{g_{\sf{sig}}}\left( t \right) = {e^{ - j\left( {\frac{{{a_0}}}{{2{b_0}}}{t^2} + \frac{{{p_0}}}{{{b_0}}}t} \right)}}\sum\nolimits_{m = 1}^{m = 3} {{\alpha _k}\cos \left( {2\pi {\omega _k}t} \right)} \]
with weight vector $\boldsymbol{\alpha}  = \left[ {35,18,10} \right]$ and frequency vector $\boldsymbol{\omega}  = \left[ {0.77,0.31,0.25} \right]$ are given. The SAFT parameter vector for the experiment is chosen to be $\left[ {{a_0},{b_0},{c_0},{d_0},{p_0},{q_0}} \right] = [7, 2, 0.6, 0.3143, 2.5, 1]$ and the sampling rate $T =\pi b_0/60$. Using (\ref{FDF}), we estimate the samples $g^{\tau}\l kT\r$ of the function $g_{\sf{sig}}^{\tau}\l t \r = g_{\textsf{sig}} \l t - \tau \r$ for $\tau  = \frac{{mT}}{{10}},m = 1, \ldots ,5.$ We compare the shift--invariant model for SAFT domain with the traditional Shannon's sampling series for SAFT ($\nu = \sinc, \vartheta [k] = \delta_k$). As before, \cite{Bhandari2012}, the metric for measuring distortion is set to be peak--signal--to--noise ratio (PSNR),
\[{\textsf{PSNR }} = 10{\log _{10}}\left( {\tfrac{{\max \left\{ {{{| {g_{{\textsf{sig}}}^\tau \left( {kT} \right)}|}^2}} \right\}}}{{{\textsf{E}}\left\{ | {{g^\tau }\left( {kT} \right)-g_{{\textsf{sig}}}^\tau \left( {kT} \right)} |^2\right\}}}} \right), \ \ \textrm{(in dB)}\]
where $\mathsf{E}\{ \cdot \}$ is the usual expectation operator. Figure~\ref{graphic1} summarizes the result of experimentation. For several choices of $\tau \in [0,T]$, the shift--invariant model for SAFT outperforms the traditional method of filtering.

\begin{figure}[!t]
\centerline{\includegraphics[width=0.6\columnwidth]{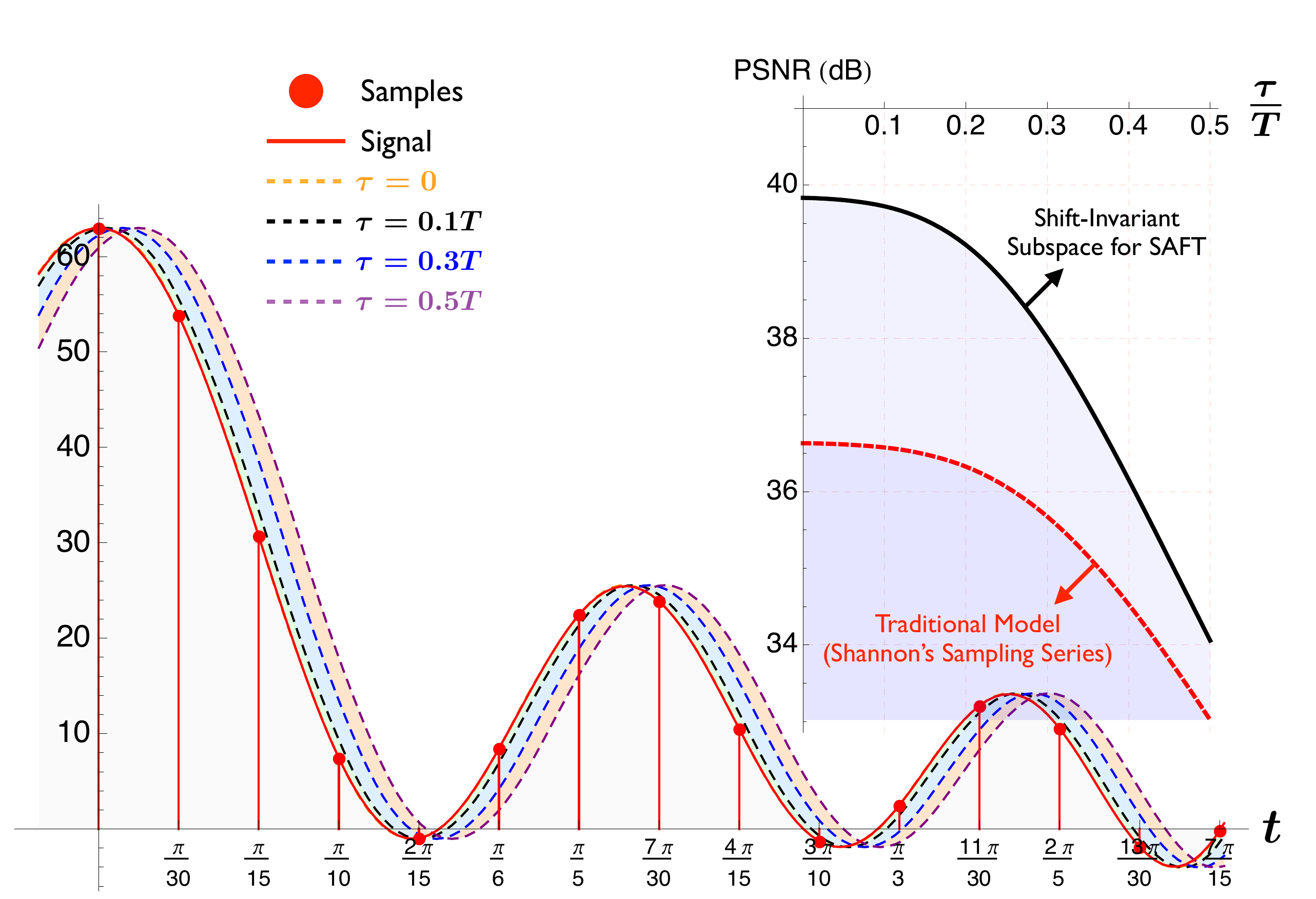}}
\caption{Fractional Delay Filtering of signals in SAFT domain using the shift--invariant signal model. We show equidistant samples of signal ${g_{\textsf{sig}}}\left( t \right) = {e^{ - j\left( {\frac{{{a_0}}}{{2{b_0}}}{t^2} + \frac{{{p_0}}}{{{b_0}}}t} \right)}}\sum\nolimits_{m = 1}^{m = 3} {{\alpha _k}\cos \left( {2\pi {\omega _k}t} \right)}$ acquired at sampling rate $T = \pi/30$. The signal is characterized by coefficient and frequency vectors $\boldsymbol{\alpha}  = \left[ {35,18,10} \right]$ and $\boldsymbol{\omega}  = \left[ {0.77,0.31,0.25} \right]$, respectively. The SAFT parameter vector is set to be $\left[ {{a_0},{b_0},{c_0},{d_0},{p_0},{q_0}} \right] = [{\text{7, 2, 0.6, 0.3143, 2.5, 1]}}$. Using $\nu \left( t \right) = \left( {2/3} \right){\cos ^4}\left( {\pi t/4} \right){\chi _{\left[ { - 2, + 2} \right]}}\left( t \right)$ as the generator of shift--invariant subspace, we estimate the $\tau$--delayed samples $g^{\tau}(kT), k\in\Z$ and compare it with the analytical samples $g^{\tau}_{\textsf{sig}}\l kT\r, k \in \Z$. Given $g_{\textsf{sig}}\l kT \r$, the choice $\tau =0$ amounts to reconstruction/interpolation the signal (see (\ref{interp},\ref{equivalent}). For other choices of $\tau = 0.1T, 0.3T \mbox{ and } 0.5T$, we reconstruct shifted versions of the $g_{\textsf{sig}}\l t \r$. (Inset) Comparison between shift--invariant model for SAFT using $\nu$ in (\ref{PC}) and the classical Shannon's sampling method using $\nu = \sinc$. The comparison metric is chosen to be peak--signal--to--noise ratio. The shift--invariant model for SAFT proves to be a better solution.}
\label{graphic1}
\end{figure}

\section{Conclusions}

We have shown that several properties and harmonic analysis results that pertain to the LCT and FrFT can be extended to the SAFT. For example, we derived  shift-invariant and reproducing-kernel Hilbert spaces associated with the SAFT, as well as, the sampling theorem, the Poisson summation formula, and the Zak transform that correspond to the SAFT. Moreover, we derived necessary and sufficient conditions for a function to be a generator for the shift-invariant space associated with the SAFT.

We would like to emphasize that our derivation of the sampling theorem of the SAFT, which is based on the theory of reproducing-kernel Hilbert spaces, is different from those used for the SAFT, LCT and FrFT; see for example, \cite{Stern2006,Stern2007,Li2007,Healy2009,Xiang2013}; hence, a new proof of the sampling theorem of the LCT may be deduced from our proof.

We have also proved that sampling in the SAFT domain is equivalent to orthogonal projection of functions onto a subspace of bandlimited basis associated with the SAFT domain. This interpretation of sampling leads to least--squares optimal sampling theorem.
Furthermore, we have shown that this approximation procedure is linked to convolution and semi--discrete convolution operators that are associated with the SAFT domain. We have presented a strategy for reconstruction/interpolation of functions using arbitrary basis functions---not necessarily $\sinc$ functions, as is the case generally---that served as an extension of the sampling theorem to generic basis functions. We concluded the article with an application of fractional delay filtering of SAFT bandlimited functions. Our ideas may be studied in context of sparse sampling theory \cite{Bhandari2016a} as well as super-resolution \cite{Bhandari2015SR}.

\bibliographystyle{IEEEtran}
\bibliography{SAFTShannon}
\end{spacing}
\end{document}